\documentclass[]{article}
\usepackage[T1]{fontenc}
\usepackage{amsmath}
\usepackage{amsfonts}
\usepackage{amssymb}

\usepackage{amsthm}

\usepackage{hyperref}
\usepackage{alltt}
\usepackage{cleveref}
\usepackage{tikz}
\usepackage[linesnumbered]{algorithm2e}
\usepackage{stmaryrd}
\usepackage{algorithmicx}
\usetikzlibrary{arrows,automata}
\usetikzlibrary{automata,positioning,decorations.pathreplacing}
\newtheorem{theorem}{Theorem}
\newtheorem{mytheorem}[theorem]{Theorem}{\bfseries}{}
\newtheorem{mydefinition}[theorem]{Definition}{\bfseries}{}
\newtheorem{mycorollary}[theorem]{Corollary}{\bfseries}{}

\newcommand{\bit}[1]{\mathtt{#1}}

\newcommand{\pijl}[1]{\mathbin{\xrightarrow{\raisebox{-0.25ex}[0cm][0cm]{\scriptsize$#1$}}}}

\newcommand{\Nat}{{\mathbb{N}}}

\newcommand{\Lang}{\mathcal{L}}
\newcommand{\simby}{\sqsubseteq}

\begin{document}
\author{
Jan Friso Groote\\[1ex]
\small Eindhoven University of Technology\\
\small Eindhoven, The Netherlands\\
\small \texttt{j.f.groote@tue.nl}\\
\and
Jan Martens \\[1ex]
\small Leiden University\\
\small Leiden, The Netherlands\\
\small \texttt{j.j.m.martens@liacs.leidenuniv.nl}
}



\title{A Quadratic Lower Bound for Simulation}

\maketitle
\begin{abstract}
    \noindent%
We show that deciding simulation equivalence and simulation preorder have quadratic lower bounds assuming
that the Strong Exponential Time Hypothesis holds. 
This is in line with the best know quadratic upper bounds of simulation equivalence. This means
that deciding simulation is inherently quadratic. 
A typical consequence of this result is that computing simulation equivalence is fundamentally harder than bisimilarity.
\end{abstract}

\section{Introduction}
Commonly, processes are abstractly represented as directed graphs with states and
transitions where states or transitions are labelled. 
One process is simulated by another if each step of the first
process can be simulated by the second. Two processes are simulation equivalent
if they can both simulate each other. Simulation preorders for the comparison of
the behaviour of programs have been defined in \cite{DBLP:conf/ijcai/Milner71}.
Algorithms for deciding simulation and simulation preorders were only defined
decades later
\cite{DBLP:journals/scp/BloomP95,DBLP:conf/focs/HenzingerHK95,DBLP:journals/tocl/BustanG03}
and they all have a complexity in terms of the number of states times the number
of transitions. In
\cite{DBLP:conf/lics/RanzatoT07,DBLP:conf/cav/GlabbeekP08,DBLP:journals/fuin/CrafaRT11,DBLP:journals/acta/Ranzato14}
these algorithms are improved by reducing both the required memory footprint and
the time complexity to the number of simulation equivalence classes times the
number of transitions, but as the number of equivalence classes can be the
number of states this essentially means that all algorithms have quadratic time
complexity in terms of the input.

This raises the question whether calculating simulation preorder and simulation equivalence is essentially 
quadratic. Hitherto, there is no real answer to this question. The only analysis is that determining 
simulation is at least as hard as bisimilarity \cite{DBLP:conf/concur/KuceraM02}. This follows from the
observation that via a polynomial translation of the process graph bisimulation equivalence reduces to simulation
equivalence, showing that modulo this transformation simulation can be used to calculate bisimilarity. 
For bisimilarity there is a quasi-linear lower bound assuming partition refinement is used 
\cite{DBLP:journals/lmcs/GrooteMV23}. 

In this paper we provide an answer by showing that if the Strong Exponential Time Hypothesis (SETH) holds,
then determining simulation preorder on deterministic transition systems and simulation equivalence
on non-deterministic transition systems must be quadratic in complexity. 

The Strong Exponential Time Hypothesis (SETH) states that satisfiability of a propositional formula with $n$ propositional
variables cannot be solved in time $O(2^{\delta n})$ for any $\delta<1$ \cite{impagliazzo2001complexity}.
There is also the Exponential Time Hypothesis (ETH), also occurring in 
\cite{impagliazzo2001complexity}, which says that determining satisfiability of
propositional formulas in conjunctive normal form
where each clause has length 3 needs at least time $O(2^{{\epsilon}n})$ for some $\epsilon>0$. The strong
exponential time hypothesis implies the exponential time hypothesis, which in turn implies $P\not=\textit{NP}$.  
Both hypotheses ETH and SETH state that in essence no algorithm for satisfiability significantly 
outperforms brute-force methods.
They are especially useful to prove lower bounds for problems within $P$~\cite{oliveira2020finegrained,williams2015hardness}.

For us the results in~\cite{oliveira2020finegrained} are particularly interesting.  
They deal with the problem of determining the non-emptiness of the 
intersection for $k$ deterministic finite state machines~(\texttt{$k$-DFA-NEI}).
This problem is defined as as
follows. Given $k$ deterministic finite state machines $A_1, \dots, A_k$, each
over an alphabet $\Sigma$ and having $n$ states, determine whether \[\bigcap_{i=1}^{k} 
\Lang(A_i) \stackrel{?}{=} \emptyset\] where $\Lang(A_i)$ is the set of accepted
strings by state machine $A_i$. For unbounded $k$ this problem is
PSPACE-complete~\cite{kozen1977lower}. If the number of input DFAs $k$ is fixed,
this problem can be solved naively in $O(n^k)$ time by constructing the product
automaton of the input.
Under the assumption ETH, it cannot be solved in $O(n^{o(k)})$ \cite[Prop.
3]{fernau2017problems}. Furthermore, 
under the stronger assumption SETH, an algorithm
running in $O(n^{k-\epsilon})$ is impossible for any constant $\epsilon > 0$,
which is more interesting to us.

We show in a quite straightforward way that if simulation preorder on the initial states of two
DFAs, all with $n$ states, can be determined
in time $O(n^{2-\epsilon})$, then $2$-\texttt{DFA-NEI} can be solved in time $O(n^{2-\epsilon})$, which
would imply that the assumption SETH does not hold. 
As a corollary it follows that computing simulation equivalence for the initial states of two 
nondeterministic automata is also inherently quadratic.






\section{SETH implies that DFA-NEI has quadratic complexity}
In this section we rephrase in more detail that the strong exponential time
hypothesis (SETH) implies that calculating the non-emptiness of the intersection
of two deterministic finite state machines requires quadratic time. In particular,
we illustrate by example how a slightly improved algorithm for the non-emptiness
of language intersection would mean an exponential improvement for
\texttt{CNF-SAT}. More detailed proofs of this construction can be found in
\cite[Theorem 7.21]{wehar2017complexity}. 

We start out with some preliminaries and the common notion of a deterministic
finite automaton.
An alphabet is a finite set of letters $\Sigma$. A word is a finite sequence of
letters over an alphabet where we write $\epsilon$ for the empty sequence. 
For a number $i\in \Nat$ the set $\Sigma^i$ is the set
of all sequences of length $i$. The set $\Sigma^* = \bigcup_{i\in\Nat} \Sigma^i$ is
the Kleene closure and contains all finite words over $\Sigma$. Given a word
$w\in\Sigma^*$ and a position $1 \leq i \leq \lvert w\rvert$, we write $w[i]$
for the $i$-th symbol in $w$.

\begin{mydefinition}
  A Deterministic Finite Automaton~(DFA) $A=(S,\Sigma, \delta, F, q_0)$ is a five-tuple
 consisting of: 
  \begin{itemize}
    \item a finite set of states $S$,
    \item a finite set of labels $\Sigma$ called the alphabet,
    \item a deterministic transition function $\delta: S\times \Sigma \mapsto S$,
    \item a set of final states $F\subseteq Q$, and
    \item an initial state $q_0\in Q$.
  \end{itemize}
\end{mydefinition}

The \emph{language} accepted by a DFA $A= (Q, \Sigma, \delta, F, q_0)$,
denoted as $\Lang(A)$, is the set of words $w\in \Sigma^*$ such that the path 
with labels from $w$ starting in $q_0$ ends up in an accepting state. 

\[
  \Lang(A) = \{ w \mid w\in\Sigma^* \text{ and } \delta(q_0, w) \in F\}
\]
where we use the generalised transition function 
$\delta$ by taking $\delta(q,\epsilon)=q$ and 
$\delta(q,a\,w)=\delta(\delta(a,q),w)$ for any state $q$. 

Given a finite number $k$
of DFAs the non-empty intersection problem of DFAs asks whether there is a word
which is accepted by all DFAs. More concretely, it contains all tuples of DFAs 
of which the intersection of the accepted languages is not empty. 
\begin{mydefinition} 
The decision problem $\texttt{$k$-DFA-NEI}$ is the following set of tuples of DFAs
\[
  \texttt{$k$-DFA-NEI} = \{ \langle A_1, \dots, A_k \rangle \mid \bigcap_{i\in[1,k]} \Lang(A_i) \, \neq \, \emptyset\}.
\]
\end{mydefinition}

We define \texttt{DFA-NEI} as the union of 
\texttt{$k$-DFA-NEI} for all $k$. This problem is well known to be PSPACE-complete~\cite{kozen1977lower}. 
For a fixed $k$ the problem is naively solvable in $O(n^k)$ by
computing the product automata. Surprisingly, it turns out that if this can be
computed more efficiently, it also means more efficient algorithms for deciding
\texttt{CNF-SAT}~\cite{wehar2017complexity}. We elaborate on that below.
\bigskip\\
\noindent%
Given an $\ell\in\Nat$, the decision problem $\ell$-\texttt{CNF-SAT} is the variant of the
boolean satisfiability problem over formulas in conjunctive normal form with at
most $\ell$ literals per clause. The computational complexity of deciding
$\ell$-\texttt{CNF-SAT}, for increasing values of $\ell$ is studied
in~\cite{impagliazzo2001complexity}. Let 
\[
 s_\ell = \inf\{\delta \mid  \text{$\ell$-\texttt{CNF-SAT} is solvable in time } 2^{\delta n}\},
\] 
for each $\ell\in\Nat$, where $\inf$ is the infimum. 
The exponential time hypothesis (ETH) asserts that $s_3 > 0$, 
i.e., $3$-\texttt{CNF-SAT} cannot be solved in less than exponential time. 
The
strong exponential time hypothesis~(SETH) asserts that 
\[
\lim_{\ell\to \infty}
s_\ell = 1.\]

Equivalently SETH asserts that for any $\delta < 1$ there is no algorithm
solving \texttt{CNF-SAT} that has a runtime of $O(2^{\delta n})$.

We consider a CNF-formula  $\Phi$ with $m$ clauses $C_1, \dots , C_m$ and an even number $n$
of propositional variables $x_1, \dots , x_n$. We construct the languages
$L^\Phi_1, L^\Phi_2 \subseteq \{\bit{0}, \bit{1}\}^*$ that consist of words $w
\in \Sigma^{n+m}$ in which the first $n$ letters comprise a 
bitstring in which for each
$1 \leq i\leq n$ the bit $w[i]$ encodes the truth assignment for $x_i$. The
second part of $m$ letters encode a gadget which assigns each clause to either
$L^\Phi_1$ or $L^\Phi_2$. We require that $w \in L^\Phi_1$ if and only if $w=
w_{\rho} b_1\cdots b_m$ for a word $w_{\rho}\in \Sigma^n$ modelling a truth
assignment and $b_1,\dots, b_m \in \{\bit{0}, \bit{1}\}$ such that if $b_i =
\bit{0}$ then $C_i$ is satisfied by the assignment to some variable $x_1, \dots,
x_{\frac{1}{2} n}$.  Similarly, $w_\rho b_1\cdots b_m \in L^\Phi_2$ if and only
for all $1\leq i \leq m$ if $b_i = 1$ then $C_i$ is satisfied by the assignment
of some variable $x_{\frac{1}{2}n+1}, \dots, x_n$. 

For example, consider $\Psi = (x_1 \vee \overline{x_2}) \wedge (\overline{x_1}
\vee x_2)$. In Fig.~\ref{fig:dfa-sat} the minimal automata that accepts
the languages $L_1^{\Psi}$ and $L_2^{\Psi}$ are given. 

Observe that, given a formula $\Phi$, for each truth assignment $w_\rho \in
\{0,1\}^n$ there is an extension of $w_C \in \{0,1\}^m$ such that $w_\rho w_c
\in L^\Phi_1 \cap L^\Phi_2$ if and only if $w_\rho$ models a satisfying
assignment for $\Phi$. This means the intersection $L^\Phi_1 \cap L^\Phi_2$ is
not empty if and only if there is a satisfying assignment for $\Phi$.

This allows us to prove satisfiability of $\Phi$ by constructing $L^\Phi_1$
and $L^\Phi_2$ and show that their intersection is not empty. Now note that the
minimal DFAs that accepts $L^\Phi_1$ and $L^\Phi_2$ both contain at most $m n
2^{\frac{1}{2} n}$ states. 

Assume that we can calculate language intersection on two graphs with $N$ states
in time $O(N^{2-\epsilon})$, then we can construct $L^\Phi_1$ and $L^\Phi_2$, and
calculate their intersection in time
$O(n^{2-\epsilon}m^{2-\epsilon}2^{n(1-\frac{1}{2}\epsilon)})$ determining \texttt{CNF-SAT}. 
But this refutes the Strong Exponential Time Hypothesis, saying that for
\texttt{CNF-SAT} cannot be solved $2^{\delta n}$ for any $\delta<1$.

%
%

\begin{figure}
  \begin{center}
  \begin{tikzpicture}[initial text=, >=stealth',semithick, node distance=1.5cm]
    \node[initial above, state] (a) {$s_0$};
    \node[state, below right=0.8cm and 0.6cm of a] (b1) {${x_1}$};
    \node[state, below left=0.8cm and 0.6cm of a] (a1) {$\overline{x_1}$};
    \node[state, below of= a1] (a2) {}; 
    \node[state, below of = a2] (a3) {};
    \node[state, below of= b1] (b2) {}; 
    \node[state, below of = b2] (b3) {};
    \node[state,accepting, below= 5cm of a] (accept) {$\top$};
    
    \node[initial above, state,right=4cm of a] (c) {$t_0$};
    \node[state, below of= c] (c1) {$ $};
    \node[state, below right=0.8cm and 0.6cm of c1] (c2) {$x_2$};
    \node[state, below left=0.8cm and 0.6cm of c1] (d2) {$\overline{x_2}$};
    \node[state, below of = c2] (c3) {};
    \node[state, below of = d2] (d3) {};
    \node[state,accepting, below= 5cm of c] (accept2) {$\top$};
  
    \path[->] (a) edge node[above right] {$\bit1$} (b1)
              (a) edge node[above left] {$\bit0$} (a1)
              (a1) edge [bend right] node[left] {$\bit0$} (a2)
              (a1) edge [bend left] node[right] {$\bit1$} (a2)
              (b1) edge [bend right] node[left] {$\bit0$} (b2)
              (b1) edge [bend left] node[right] {$\bit1$} (b2)
              (a2) edge node [left] {$\bit 1$} (a3)
              (b2) edge [bend right] node[left] {$\bit0$} (b3)
              (b2) edge [bend left] node[right] {$\bit1$} (b3)
              (b3) edge [] node[below right] {$\bit1$} (accept)
              (a3) edge [bend right] node[below left] {$\bit0$} (accept)
              (a3) edge [bend left] node[above right] {$\bit1$} (accept)
              (c) edge [bend right] node[left] {$\bit0$} (c1)
              (c) edge [bend left] node[right] {$\bit1$} (c1)
              (c1) edge node [above left] {$\bit 0$} (d2)
              (c1) edge node [above right] {$\bit 1$} (c2)
              (c2) edge node[right] {$\bit0$} (c3)
              (d2) edge [bend right] node[left] {$\bit0$} (d3)
              (d2) edge [bend left] node[right] {$\bit1$} (d3)
              (d3) edge [] node[below left] {$\bit0$} (accept2)
              (c3) edge [bend right] node[above left] {$\bit0$} (accept2)
              (c3) edge [bend left] node[below right] {$\bit1$} (accept2)
              ;

  \end{tikzpicture}
  \end{center}
  \caption{Automata accepting $L^\Psi_1$ (left), and $L^\Psi_2$ (right) for 
  $\Psi=(x_1 \vee \overline{x_2}) \wedge (\overline{x_1}
\vee x_2)$.\label{fig:dfa-sat}}
  \end{figure}
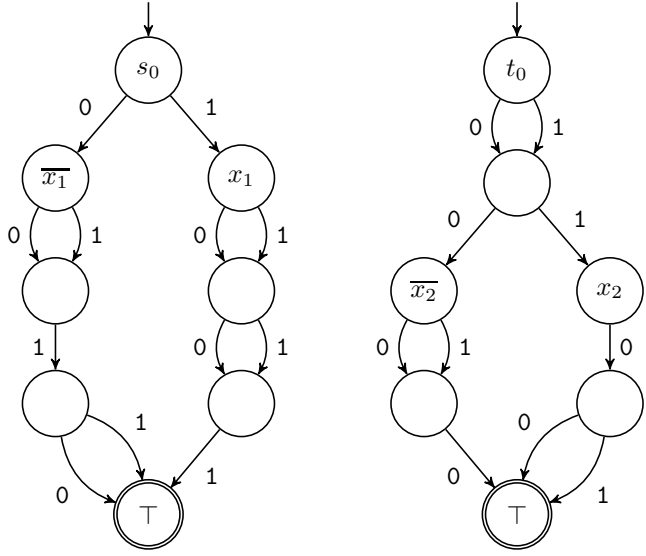

  


\begin{mytheorem}{\cite[Theorem 7.21]{wehar2017complexity}}
 If $2$-\texttt{DFA-NEI} can be solved in $O(n^{2-\epsilon})$ for some constant
 $\epsilon >0$, then SETH is false.
\end{mytheorem}

\section{SETH implies simulation has quadratic complexity}
In this section we show that determining simulation preorder on deterministic, and simulation equivalence 
on non-deterministic transition systems is necessarily quadratic, assuming SETH is valid. 
Simulation is typically defined on labelled transition systems which divert slightly from DFAs. 
\begin{mydefinition}%
   A Labelled Transitions System~(LTS) $M = (S, \Sigma, \pijl{}, s_0)$ is a
   four-tuple consisting of:
   \begin{itemize}
   \item a finite set of states $S$,
   \item a finite set of action labels $\Sigma$,
   \item a transition relation $\pijl{} \subseteq S\times \Sigma \times S$, and
   \item the initial state $s_0 \in S$.
   \end{itemize}
\end{mydefinition}

Given an LTS $L= (S,Act, \pijl{})$, we write $s\pijl{a}s'$ when $(s,a,s') \in
\pijl{}$. The LTS $L$ is called \textit{deterministic} iff there is at most
one outgoing transition for every combination of state $s\in S$ and action label
$a\in Act$, i.e. $\lvert \{ s' \mid s\pijl{a}s' \}\rvert \leq 1$ for every
$(s,a) \in S\times \Sigma$.

\begin{mydefinition}
  Given an LTS $M=(S,\Sigma, \pijl{}, s_0)$, a relation $R \subseteq S\times S$
  is called a \textit{simulation relation} iff for all $(s,t) \in R$ it holds that
  \begin{itemize}
    \item for each  $s\pijl{a} s'$ there is a transition $t\pijl{a}t'$ such that
    $(s',t') \in R$.
  \end{itemize}
\end{mydefinition}

The largest simulation relation, written as $\simby$, is called similarity. Given
an LTS $M_1 = (S_1, Act_1, \pijl{}_1, s_1)$ and two states $s,t\in S_1$ we say
$s$ is simulated by $t$ iff $s \simby t$. Given a second LTS $M_2 = (S_2,
Act_2, \pijl{}_2,s_2)$, we write $M_1 \simby M_2$ iff $s_1\simby s_2$ in the
combined LTS $M = (S_1 \cup S_2, Act_1 \cup Act_2, \pijl{}_1 \cup
\pijl{}_2, s_1)$, where we assume w.l.o.g. that the states of $M_1$ and $M_2$
are disjoint, e.g. $S_1 \cap S_2 = \emptyset$. 

We say two states $s,t$ are \textit{simulation equivalent}, written $s\simeq t$, iff they
simulate each other. A simulation relation which is symmetric is called a
bisimulation relation.

We introduce a mapping $\alpha$ from automata to deterministic LTSs such that given two DFAs
$A_1, A_2$, for the projected deterministic LTSs $M_1 = \alpha(A_1), M_2 = \alpha(A_2)$ it
holds that $M_1 \simby M_2$ if and only if $\Lang(A_1)\subseteq \Lang(A_2)$.

Informally, the mapping $\alpha$ maintains the same transition structure in the
deterministic LTSs but adds one state and action label that encodes the accepting states of
the automata. More formally, given a DFA $A=(Q,\Sigma, \delta, F, q_0)$, we
define the deterministic LTS $\alpha(A) = (Q \cup \{\top\}, \Sigma \cup \{\checkmark\},
\pijl{}, q_0)$, with a fresh symbols $\top\not\in Q,\checkmark\not\in \Sigma$,
and where the transition relation $\pijl{}$ is defined as:
\begin{align*}
{\pijl{}} = \{ (q,a,\delta(q,a)) \mid \text{for each } q,a\in Q\times \Sigma\}\, \cup\, 
\{(q,\checkmark, \top) \mid \textrm{for each } q\in F\}.
\end{align*}

As the construction above is quite straightforward, it is easy to see that we have the following theorem.
\begin{mytheorem}\label{thm:subset-sim}
 Let $A,B$ be DFAs over the alphabet $\Sigma$, then 
 \[ 
  \Lang(A) \subseteq \Lang(B) \iff \alpha(A) \simby \alpha(B)
 \]
\end{mytheorem}

The construction $\alpha(A)$ is computable in linear time. This allows us to
compute \texttt{$2$-DFA-NEI} by translating the involved DFAs to deterministic LTSs using
$\alpha$.

\begin{mytheorem}
  If for some time function $f$ similarity on two deterministic LTSs with $n$ states can be decided in $f(n)$ steps then
  $2\texttt{-DFA{-}NEI}$ for input DFAs of $n$ states is computable in $f(n) +
  O(n)$ steps.
\end{mytheorem}
\begin{proof}
  Not that deciding $\Lang(A) \cap \Lang(B) = \emptyset$ is equivalent to
deciding $\Lang(A) \subseteq \Sigma^*\setminus \Lang(B)$. The DFA $\overline{B}$
is the complement of $B$, e.g. with all accepting and non-accepting states
swapped, such that $\Lang(\overline{B}) = \Sigma^* \setminus \Lang(B)$.

Now by Theorem~\ref{thm:subset-sim} it holds that $\alpha(A) \simby
\alpha\left(\overline{B}\right)$ if and only if $\Lang(A) \cap \Lang(B) =
\emptyset$. Computing $\alpha(A),\alpha(\overline{B})$ can be done in $O(n)$,
and hence any $f(n)$ algorithm for similarity could be translated to
$2\texttt{-DFA-NEI}$.
\end{proof}

This immediately translates to the following corollary.

\begin{mycorollary}\label{cor:similarity}
  SETH implies that similarity for deterministic LTSs can not be decided in $O(n^{2-\epsilon})$. 
\end{mycorollary}

On deterministic structures simulation equivalence coincides with bisimilarity.
Meaning that deciding simulation equivalence on deterministic structures can be
done in almost linear
time~\cite{hopcroft1971linear,DBLP:conf/coco/Fischer72,tarjan1975efficiency},
where `almost' refers to a multiplicative factor with the inverse Ackermann's
function. However, using one non-deterministic transition a faster than
quadratic algorithm for simulation equivalence on LTSs in general violates SETH.

\begin{mycorollary}
  SETH implies that simulation equivalence cannot be decided in $O(n^{2-\epsilon})$.
\end{mycorollary}
\begin{proof}
  In order to see that deciding simulation equivalence is also computationally
  equivalent to deciding $2\texttt{-DFA-NEI}$ we reduce from deciding similarity
  on deterministic LTSs by adding one non-deterministic state. Given LTSs $M_1 = (S_1,
  \Sigma, \pijl{}_1, s_0)$ and $M_2 = (S_2,  \Sigma, \pijl{}_2, t_0)$, we
  construct the LTS $M = (S_1 \cup S_2 \cup \{s,t\}, \pijl{}_1 \cup \pijl{}_2
  \cup \{(s ,a, s_0), (s,a,t_0), (t,a,t_0)\}$ with two fresh states $s,t\not\in
  S_1 \cup S_2$ for some $a\in\Sigma$. See Fig.~\ref{fig:M}. Now in $M$ it holds that $s \simeq t \iff
  s_0 \simby t_0$. If simulation equivalence for LTSs is solvable in
  $O(n^{2-\epsilon})$, then, using this construction, similarity is solvable in
  $O(n^{2-\epsilon})$. In that case it follows from Corollary~\ref{cor:similarity} that SETH is false.
\end{proof}

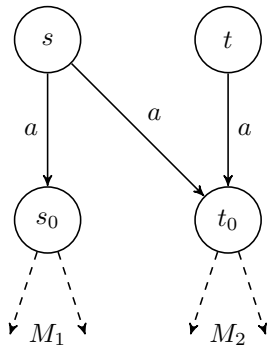
\begin{figure}
  \centering
\begin{tikzpicture}[initial text=, >=stealth',semithick, node distance=1.5cm]
  \node[state] (s) {$s$};
  \node [state, right= of s] (t) {$t$};

  \node[state,below= of s] (m1) {$s_0$};
  \node[state, below= of t] (m2) {$t_0$};
  \node[below= 0.8cm of m1] (a) {$M_1$};
  \node[below= 0.8cm of m2] (b) {$M_2$};

  \path[->] (s) edge node[left] {$a$} (m1)
            (t) edge node[right] {$a$} (m2)
            (s) edge node[above right] {$a$} (m2);
            
  \path[->, dashed] (m1) edge ([xshift=0.5cm] a)
            (m1) edge ([xshift=-0.5cm] a)
            (m2) edge ([xshift=0.5cm] b)
            (m2) edge ([xshift=-0.5cm] b);     
\end{tikzpicture}
\caption{The LTS $M$.  \label{fig:M}}
\end{figure}

The proof technique employed in this paper appears to be quite universal. It
promises to be useable to show fine-grained, i.e., sub-exponential, lower bounds
for a much wider classes of problems in the domain of process theory and model
checking. Unfortunately, as yet, we were not able to do so.

\bibliographystyle{plainurl}
\bibliography{bibliography}
\end{document}